\documentclass[runningheads,a4paper]{llncs}

\usepackage{amssymb,amsmath,stmaryrd}
\usepackage{graphicx}
\usepackage{epstopdf}
\usepackage{longtable}
\usepackage{multirow}
\usepackage{fancyvrb}
\usepackage{bibentry}
\usepackage{bm}
\usepackage{setspace}
\usepackage{algorithm}
\usepackage{algorithmic}
\usepackage[english]{babel}
\usepackage{listings}
\usepackage{xcolor}

\newtheorem{mydefinition}{Definition}

\newtheorem{myexample}{Example}

\allowdisplaybreaks

\usepackage{url}
\urldef{\mailsa}\path|xulu@stu.mail.xidian.edu.cn|
\urldef{\mailsb}\path|zhhduan@mail.xidian.edu.cn|
\urldef{\mailsc}\path|tico_tools@163.com|
\urldef{\mailsd}\path|lhjbuaa@hotmail.com|
\newcommand{\keywords}[1]{\par\addvspace\baselineskip
\noindent\keywordname\enspace\ignorespaces#1}

\begin{document}

\mainmatter
\title{Extending PPTL for Verifying Heap Evolution Properties}  % start of an individual contribution

%\title{Integrating Separation Logic with PPTL}

\titlerunning{Extending PPTL for Verifying Heap Evolution Properties}

\author{Xu Lu\and Zhenhua Duan\and Cong Tian}
\authorrunning{Xu Lu\and Zhenhua Duan\and Cong Tian}

\institute{ICTT and ISN Lab, Xidian University, Xi'an, 710071, P.R. China}


\toctitle{Lecture Notes in Computer Science}
\tocauthor{Authors' Instructions}
\maketitle

\begin{abstract}
In this paper, we integrate separation logic with Propositional Projection Temporal Logic (PPTL) to obtain a two-dimensional logic, namely PPTL$^{\tiny\mbox{SL}}$. The spatial dimension is realized by a decidable fragment of separation logic which can be used to describe linked lists, and the temporal dimension is expressed by PPTL. We show  that PPTL and PPTL$^{\tiny\mbox{SL}}$ are closely related in their syntax structures. That is, for any PPTL$^{\tiny\mbox{SL}}$ formula in a restricted form, there exists an ``isomorphic" PPTL formula. The ``isomorphic" PPTL formulas can be obtained by first an equisatisfiable translation and then an isomorphic mapping. As a result, existing theory of PPTL, such as decision procedure for satisfiability and model checking algorithm,  can be reused for PPTL$^{\tiny\mbox{SL}}$.
\end{abstract}

\section{Introduction}

The heap is an area of memory for dynamic memory allocation and pointers are references to heap cells. It is hard to detect errors of heap-manipulating programs with inappropriate management of heap. Verification of such programs is an active research field today and has had a long history ever since the early 1970s \cite{Some Techniques for Proving Correctness of Programs which Alter Data Structures}. However, it is still a big challenge because of aliasing \cite{A Trace Model for Pointers and Objects}. Programs become more error-prone with serious problems, e.g. the existence of memory violation, the emergence of memory leaks, etc.. In addition, reasoning about temporal properties about the heap of these programs is even more difficult than just memory safety properties.

Reynolds \cite{Separation Logic: A Logic for Shared Mutable Data Structures} and O'Hearn \cite{BI as an Assertion Language for Mutable Data Structures} proposed a Hoare-style logic which is known as separation logic. More recently, separation logic is increasingly being used and extended for automated assertion checking \cite{Symbolic Execution with Separation Logic} and shape analysis \cite{A Local Shape Analysis based on Separation Logic}. Although separation logic is very popular due to its reasoning power in heap-manipulating programs, we emphasize that it is a variant of Hoare-like proof systems. That is to say, it is a static logic which infers assertions at each program point. What we mean by ``static'' is that separation logic is short in the power for expressing heap evolution properties which can be seen ``dynamic'' heap evolutions by time. For instance, separation logic formula $\phi \text{\scriptsize{\#}} \phi'$ specifies properties, at one state, holding respectively for disjoint portions of the heap, one makes $\phi$ true and the other makes $\phi'$ true. But temporal property like the heap can be divided into two disjoint sub-heaps ($\phi \text{\scriptsize{\#}} \phi'$) always (or eventually) holds during program execution cannot be expressed by separation logic.

Temporal logic is another highly successful formalism which has already been well-developed in automatic program verification. There are various versions of temporal logic such as Computation Tree Logic (CTL) \cite{The Temporal Logic of Branching Time} and Linear Temporal Logic (LTL) \cite{Temporal Logic of Reactive and Concurrent Systems}. While LTL is interpreted over an infinite sequence of states and CTL over a tree structure, the base logic utilized in this paper, Propositional Projection Temporal Logic (PPTL) \cite{An Extended Interval Temporal Logic and A Framing Technique for Temporal Logic Programming}, is an interval based temporal logic, which is interpreted over finite or infinite intervals. It is more powerful than both LTL and CTL with respect to expressiveness \cite{Propositional Projection Temporal Logic Buchi Automata and omega-Regular Expressions}. However, temporal logics do not have the ability to reason about heaps. Both the models and logics need to be augmented with heap ingredients if we want to deal with heaps.

%\vspace{-0.88pt}

It is useful to integrate the two types (spatial and temporal) of logics such that heap evolution properties can be specified and verified in a unified manner. There are various temporal logics previously designed for heap verification in the literature. Evolution Temporal Logic (ETL) \cite{Verifying Temporal Heap Properties Specified via Evolution Logic} is a first-order LTL for the description of program behaviors that causes dynamic allocation and deallocation of heap. ETL mainly focuses on describing large granularity heap objects and high-level threads. Based on a tableau model checking algorithm, Navigation Temporal Logic (NTL) \cite{Safety and Liveness in Concurrent Pointer Programs} extends LTL with pointer assertions for reasoning about the evolution of heap cells. Rieger established an expressive Temporal Pointer Logic (TPL) \cite{Verification of Pointer Programs} which expresses properties of computation paths and pointer comparisons evaluated on single heap states separately. An approach based on abstraction technique for TPL has to be built as the logic is in general undecidable. In \cite{Model Checking Dynamic Memory Allocation in Operating Systems}, LTL and CTL are combined, in time and space, to specify complex properties of programs with dynamic heap structures. Though it is a two-dimensional logic, both dimensions are realized by temporal logics that makes the difference between two dimensions less obvious. The work \cite{Reasoning About Sequences of Memory States} proposed by Brochenina et al. devises a logic by means of a quantifier-free fragment of separation logic as the underlying assertion language on top of which is Propositional LTL (PLTL). Formulas in this logic include pointer arithmetic that is well studied and expressions are enriched with $\bigcirc$ operator denoting the next value of them. Various classes of models and fragments of separation logic are explored in depth. Yet some common properties like memory safety and shape properties are unable to be characterized as quantifiers are not contained in any fragment, and no tool or experimental results for the logic are available yet.

%\vspace{-10pt}

In this paper, we propose a two-dimensional (spatial and temporal) logic named PPTL$^{\tiny\mbox{SL}}$ for specifying heap evolution properties of programs by integrating separation logic with PPTL. On one hand, our logic inherits the advantages of separation logic in describing heaps in a much simpler and more intuitionistic way. Meanwhile, the fragment of separation logic utilized here can describe complex heap structures. On the other hand, PPTL is more powerful than PLTL since PPTL describes full regular language \cite{Propositional Projection Temporal Logic Buchi Automata and omega-Regular Expressions}. In contrast, PLTL describes star free regular language \cite{On the Temporal Analysis of Fairness,Counter-Free Automata}. Moreover, the spatial-temporal logic introduced in this paper contains temporal connectives $``;"$ for sequentially combining formulas which enable us easily to express the occurrence of sequential events, and ``$^{+}"$ or ``$^{*}"$ enables us to state loop properties. For instance, the formula $P_{1};P_{2}$ asserts that $P_{1}$ holds from now until some point in the future, and from that point on, $P_{2}$ holds. $P^{*}$ or $P^{+}$ means that $P$ repeatedly holds for a finite or infinite number of times.

The main contribution of this paper includes: $(a)$ We propose an expressive temporal logic composed of a temporal dimension (evolution of programs) and a spatial dimension (heap structures); $(b)$ An isomorphic relationship is established between PPTL$^{\tiny\mbox{SL}}$ and PPTL in order to solve the satisfiability problem of PPTL$^{\tiny\mbox{SL}}$ and further to obtain the corresponding decision procedure. Our previous work \cite{Integrating Separation Logic with PPTL} also presents a logic that integrates separation logic with PPTL. However, a different fragment of separation logic without quantifiers is employed in that work. Therefore, the logic is limited in expressing some useful properties, e.g., memory safety properties. We only prove a normal form of the logic whose satisfiability problem remains unsolved in that paper. We believe that the direct normal form approach is not enough to guarantee the decidability of the logic. Therefore, we use an alternative approach which builds an isomorphic relationship in order to solve the satisfiability problem.

The remainder of this paper is organized as follows. In the following section, the syntax and semantics of the two-dimensional logic PPTL$^{\tiny\mbox{SL}}$ is presented. In Section \ref{sec: formula relationship}, an isomorphic relationship between PPTL and PPTL$^{\tiny\mbox{SL}}$ formulas is obtained. As a result, how the decision procedure for checking satisfiability of PPTL to be reused on PPTL$^{\tiny\mbox{SL}}$ is illustrated Section \ref{sec: Normal Form and NFG}. Conclusions are drawn in Section \ref{sec: conclusion}.

\section{The Two-Dimensional Logic PPTL$^{\tiny\mbox{SL}}$}\label{sec: hybrid logic}

The satisfiability problem for full separation logic is known to be undecidable \cite{Computability and Complexity Results for a Spatial Assertion Language for Data Structures}. In this section, we first introduce a decidable fragment of separation logic (SL for short) which is able to describe linked list structures. Then we make a temporal extension to SL by adding specific temporal operations in PPTL.

\subsection{A Fragment of Separation Logic for Linked Lists}\label{separation logic}

The fragment of separation logic presented here is a variation of the one in \cite{On the Almighty Wand}. We assume a countable infinite set $Var$ of variables with a fixed ordering, ranged over by $x, y, z, \ldots$. Let $Loc$ be a finite set of valid locations composed of the natural numbers greater than zero. $Val=Loc \cup \{\, 0 \,\}$ denotes the set of values which are either locations or $0$. The constant $0$ represents the null location. We refer to a pair $(I_{s},I_{h})$ as a memory state $s$, where $I_{s} \, : \, Var \rightharpoonup Val$ represents a stack and $I_{h} \, : \, Loc \rightharpoonup Val$ a heap.

\subsubsection{Syntax} Formulas of SL are defined by the grammar below, $n$ is a natural number:
%\begin{longtable}{lrcllrcl}
%    \mbox{Terms } &$e$ &$\, ::= \,$ &$nil \mid x$ &\quad\mbox{simple Pures } &$\alpha$ &$\, ::= \,$ &$e_{1}=e_{2} \mid \neg \alpha$ \\
%    \mbox{Pures } &$\beta$ &$::=$ &$true \mid \beta \wedge \alpha$ &\quad\mbox{simple Spatials } &$\gamma$ &$::=$ &$e_{1} \mapsto e_{2} \mid ls^{n}(e_{1},e_{2}) \mid$ \\
%    &&&&&&&$ls(e_{1},e_{2})$ \\
%    \mbox{Spatials } &$\delta$ &$::=$ &$emp \mid \delta * \gamma$ &\quad\mbox{Symbolic heaps } &$\varphi$ &$::=$ &$\beta \wedge \delta$
%\end{longtable}

\begin{longtable}{lrcl}
    \mbox{Terms} &$e$ &$\, ::= \,$ &$n \mid x$ \\
    SL Formulas &\quad$\phi$ &$\; ::= \;$ &$e_{1}=e_{2} \mid e_{1} \mapsto e_{2} \mid \neg \phi \mid \phi_{1} \vee \phi_{2} \mid \phi_{1} \text{\scriptsize{\#}} \phi_{2} \mid \exists x : \phi$
\end{longtable}
\noindent %
We will make use of standard notations as usual for other derived connectives. We write $dom(f)$ to denote the domain of mapping $f$. Given two mappings $f_{1}$ and $f_{2}$, $f_{1} \perp f_{2}$ means $f_{1}$ and $f_{2}$ with disjoint domains. Moreover, we use $f_{1} \cdot f_{2}$ to denote the union of $f_{1}$ and $f_{2}$ which is undefined when $f_{1} \not \perp f_{2}$. Formula $e_{1} \mapsto e_{2}$ denotes that $e_{1}$ points to $e_{2}$, where $e_{1}$ represents an address in the heap and $e_{2}$ the value held in that address.

\subsubsection{Semantics} For every term $e$, the evaluation of $e$ relative to a state $(I_{s},I_{h})$ is defined as $(I_{s},I_{h})[e]$.
{%\small
\begin{longtable}{rclcl}
    \multicolumn{5}{c}{$(I_{s},I_{h})[n] = n \qquad (I_{s},I_{h})[x] = I_{s}(x)$}
\end{longtable}
}
\noindent %
The semantics of SL formulas is given below by a relation $\models_{_{SL}}$ equipped with a subscript $SL$.
{%\small
\begin{longtable}{rclcl}
    $I_{s},I_{h}$ &$\models_{_{SL}}$ &$e_{1}=e_{2}$ &\;iff\;  &$(I_{s},I_{h})[e_{1}] = (I_{s},I_{h})[e_{2}]$. \\
    $I_{s},I_{h}$ &$\models_{_{SL}}$ &$e_{1} \mapsto e_{2}$ &\;iff\;  &$dom(I_{h})=\{\, (I_{s},I_{h})[e_{1}] \,\}$ and $I_{h}((I_{s},I_{h})[e_{1}]) = (I_{s},I_{h})[e_{2}]$. \\
    $I_{s},I_{h}$ &$\models_{_{SL}}$ &$\neg \phi$ &\;iff\;  &$I_{s},I_{h} \not\models_{_{SL}} \phi$. \\
    %$I_{s},I_{h}$ &$\models_{_{SL}}$ &$\neg \alpha$ &iff  &$I_{s},I_{h} \not\models_{_{SL}} \alpha$. \\
    $I_{s},I_{h}$ &$\models_{_{SL}}$ &$\phi_{1} \vee \phi_{2}$ &iff  &$I_{s},I_{h} \models_{_{SL}} \phi_{1} \text{ or } I_{s},I_{h} \models_{_{SL}} \phi_{2}$. \\
    $I_{s},I_{h}$ &$\models_{_{SL}}$ &$\phi_{1} \text{\scriptsize{\#}} \phi_{2}$ &iff  &$\text{there exist } I_{h_{1}}, I_{h_{2}} : I_{h_{1}} \perp I_{h_{2}} \text{ and } I_{h} = I_{h_{1}} \cdot I_{h_{2}} \text{ and }$ \\
    &&& &$I_{s},I_{h_{1}} \models_{_{SL}} \phi_{1} \text{ and } I_{s},I_{h_{2}} \models_{_{SL}} \phi_{2}$. \\
    $I_{s},I_{h}$ &$\models_{_{SL}}$ &$\exists x : \phi$ &iff  &$\text{there exists } v \in Val \text{ such that } I_{s}[x \rightarrow v], I_{h} \models_{_{SL}} \phi$.
\end{longtable}
}

\subsubsection{Derived formulas} We can also present the following derived formulas which can be seen as a series of useful properties expressed in SL.

\begin{longtable}{rcl}
    $e_{1} \hookrightarrow e_{2}$ &$\overset{\scriptsize\mbox{def}}{=}$& $e_{1} \mapsto e_{2} \text{\scriptsize{\#}} true \quad\;\; alloc(e) \overset{\scriptsize\mbox{def}}{=} \exists x : e \hookrightarrow x$  \\
    %$alloc(e)$ &$\, \overset{\scriptsize\mbox{def}}{=}$ \,& $\exists x : e \hookrightarrow x$ \\
    $emp$ &$\overset{\scriptsize\mbox{def}}{=}$& $\neg \exists x : alloc(x)$ \\
    $\sharp e \geq n$ &$\overset{\scriptsize\mbox{def}}{=}$& $e \neq 0 \wedge \overset{n\text{ times}}{\overbrace{(\exists y : y \mapsto e) \text{\scriptsize{\#}} \cdots \text{\scriptsize{\#}} (\exists y : y \mapsto e)}} \text{\scriptsize{\#}} true$ \\
    $\sharp e \leq n$ &$\overset{\scriptsize\mbox{def}}{=}$& $e \neq 0 \wedge \neg \big(\overset{n+1\text{ times}}{\overbrace{(\exists y : y \mapsto e) \text{\scriptsize{\#}} \cdots \text{\scriptsize{\#}} (\exists y : y \mapsto e)}} \text{\scriptsize{\#}} true \big)$ \\
    $\sharp e = n$ &$\overset{\scriptsize\mbox{def}}{=}$& $e \neq 0 \wedge \big( \overset{n\text{ times}}{\overbrace{(\exists y : y \mapsto e) \text{\scriptsize{\#}} \cdots \text{\scriptsize{\#}} (\exists y : y \mapsto e)}} \text{\scriptsize{\#}} true \big) \wedge$ \\
    &&$\neg \big(\overset{n+1\text{ times}}{\overbrace{(\exists y : y \mapsto e) \text{\scriptsize{\#}} \cdots \text{\scriptsize{\#}} (\exists y : y \mapsto e)}} \text{\scriptsize{\#}} true \big)$ \\
    $e_{1} \overset{\circlearrowleft}{\longrightarrow}^{+} e_{2}$ &$\, \overset{\scriptsize\mbox{def}}{=}$ \,& $alloc(e_{1}) \wedge (e_{2} \neq e_{1} \rightarrow \neg alloc(e_{2}) \wedge \sharp e_{1}=0) \wedge$ \\
    &&$(\forall x : x \neq e_{2} \rightarrow (\sharp x=1 \rightarrow alloc(x))) \wedge$ \\
    &&$(\forall x : x \neq 0 \rightarrow \sharp x \leq 1)$ \\
    $ls(e_{1},e_{2})$ &$\, \overset{\scriptsize\mbox{def}}{=}$ \,& $e_{1} \overset{\circlearrowleft}{\longrightarrow}^{+} e_{2} \wedge \neg (e_{1} \overset{\circlearrowleft}{\longrightarrow}^{+} e_{2} \text{\scriptsize{\#}} \neg emp)$
\end{longtable}
\noindent %
Formula $e_{1} \hookrightarrow e_{2}$ has a weaker meaning than $e_{1} \mapsto e_{2}$ since the domain of the heap of the former may contains other allocated heap cells in addition to $e_{1}$. $alloc(e)$ indicates that the cell $e$ is allocated in the current heap. $emp$ is true just for the empty heap whose domain is $\emptyset$. $\sharp e \geq n$ holds in case that $e$ has at least $n$ predecessors. $\sharp e = n$ and $\sharp e \leq n$ can be obtained by obvious combinations of comparison predicates. A state $(I_{s}, I_{h})$ satisfies $e_{1} \overset{\circlearrowleft}{\longrightarrow}^{+} e_{2}$ indicating that $I_{h}$ can be decomposed as a list segment between $e_{1}$ and $e_{2}$ and a finite collection of cyclic lists. In addition, $ls(e_{1},e_{2})$ describes a list segment starting at the location denoted by $e_{1}$ whose last link contains the value of $e_{2}$, in particular, $ls(e,0)$ is a complete linked list and $ls(e,e)$ is a cyclic linked list.

\subsection{Temporal Extension to Separation Logic}

In order to express temporal properties of heap systems, we integrate SL with PPTL. The two-dimensional logic is named as PPTL$^{\tiny\mbox{SL}}$. Let $Prop$ be a countable set of atomic propositions. Formulas $Q$ of PPTL and $P$ of PPTL$^{\tiny\mbox{SL}}$ are given by the following grammar, respectively,
\begin{longtable}{lrcl}
    PPTL Formulas &\quad$Q$ &$\; ::= \;$ &$q \mid \neg Q \mid Q_{1} \vee Q_{2} \mid {\bigcirc}Q \mid (Q_{1},\ldots,Q_{m}) \, prj \, Q \mid Q^{*}$ \\
    PPTL$^{\tiny\mbox{SL}}$ Formulas &\quad$P$ &$\; ::= \;$ &$\phi \mid \neg P \mid P_{1} \vee P_{2} \mid {\bigcirc}P \mid (P_{1},\ldots,P_{m}) \, prj \, P \mid P^{*}$
\end{longtable}
%\vspace{-5pt}
\noindent %
where $q \in Prop$, $\phi$ denotes SL formulas and $P_{1},\ldots,P_{m}$ are all well-formed PPTL$^{\tiny\mbox{SL}}$ formulas ($Q_{1},\ldots,Q_{m}$ are all well-formed PPTL formulas). $\bigcirc$ (next), $prj$ (projection) and $^{*}$ (star) are basic temporal operators. A formula is called a state formula if it does not contain any temporal operators, otherwise it is a temporal formula.

An interval $\sigma = \langle s_{0},s_{1},\ldots \rangle$ is a sequence of states, possibly finite or infinite. $\epsilon$ denotes an empty interval. The length of $\sigma$, denoted by $|\sigma|$, is $\omega$ if $\sigma$ is infinite, otherwise it is the number of states minus one. To have a uniform notation for both finite and infinite intervals, we will use extended integers as indices. That is, we consider the set $N_{0}$ of non-negative integers and $\omega$, define $N_{\omega}=N_{0} \cup \{\, \omega \,\}$, and extend the comparison operators, $=$, $<$, $\leq$, to $N_{\omega}$ by considering $\omega=\omega$, and for all $i \in N_{0}$, $i < \omega$. Moreover, we define $\preceq$ as $\leq - \{\, (\omega,\omega) \,\}$. With such a notation, $\sigma_{(i .. j)}(0 \leq i \preceq j \leq |\sigma|)$ denotes the sub-interval $\langle s_{i},\ldots,s_{j} \rangle$ and $\sigma^{(k)}(0 \leq k \preceq |\sigma|)$ denotes the suffix interval $\langle s_{k},\ldots,s_{|\sigma|} \rangle$ of $\sigma$. The concatenation of $\sigma$ with another interval (or empty string) $\sigma'$ is denoted by $\sigma \cdot \sigma'$. Further, let $\sigma = \langle s_{k},\ldots,s_{|\sigma|} \rangle$ be an interval and $r_{1},\ldots,r_{h}$ be integers $(h \geq 1)$ such that $0 \leq r_{1} \leq r_{2} \leq \cdots \leq r_{h} \preceq |\sigma|$. The projection of $\sigma$ onto $r_{1},\ldots,r_{h}$ is the interval (called projected interval), $\sigma \downarrow (r_{1},\ldots,r_{h}) = \langle s_{t_{1}}, \ldots, s_{t_{l}} \rangle$, where $t_{1},\ldots,t_{l}$ is obtained from $r_{1},\ldots,r_{h}$ by deleting all duplicates. That is, $t_{1},\ldots,t_{l}$ is the longest strictly increasing subsequence of $r_{1},\ldots,r_{h}$. For example,
\begin{eqnarray}
    \langle s_{0},s_{1},s_{2},s_{3},s_{4} \rangle \downarrow (0,0,2,2,2,3) = \langle s_{0}, s_{2}, s_{3} \rangle \nonumber
\end{eqnarray}

%The binary operators interval concatenation $\cdot$ acting on intervals and yielding a combined fresh interval is defined as follows:

%\textbf{-----}  Interval concatenation $\cdot$
%\begin{equation}
%    \sigma \cdot \sigma' =
%        \begin{cases}
%            \sigma &\mbox{if $|\sigma|=\omega$ or $\sigma'=\epsilon$}, \nonumber\\
%            \sigma' &\mbox{if $\sigma=\epsilon$}, \nonumber\\
%            \langle s_{0},\ldots,s_{i},s_{i+1},\ldots \rangle &\mbox{if $\sigma = \langle s_{0},\ldots,s_{i} \rangle$ and $\sigma' = \langle s_{i+1},\ldots \rangle$}. \nonumber\\
%%            &\mbox{$\sigma' = \langle s_{i+1},\ldots \rangle$}.
%        \end{cases}
%\end{equation}

An interpretation for a PPTL$^{\tiny\mbox{SL}}$ formula is a triple $\mathcal{I} = (\sigma, k, j)$ where $\sigma = \langle s_{0},s_{1},\ldots\rangle$ is an interval, $k$ a non-negative integer and $j$ an integer or $\omega$ such that $0 \leq k \preceq j \leq |\sigma|$. We write $(\sigma,k,j) \models P$ to mean that a formula $P$ is interpreted over a sub-interval $\sigma_{(k .. j)}$ of $\sigma$ with the current state being $s_{k}$. The notation $s_{k}=(I^{k}_{s},I^{k}_{h})$ indexed by $k$ represents the $k$-th state of an interval $\sigma$. The satisfaction relation for PPTL$^{\tiny\mbox{SL}}$ formulas $\models$ is defined as follows.
{\small
\begin{longtable}{rclcl}
    $\mathcal{I}$ &$\models$ &$\phi$ &$\text{ iff }$ &$I^{k}_{s},I^{k}_{h} \models_{_{SL}} \phi$. \\
%    &&&& $\mathcal{I}(e_{i}) \text{ is defined for all } i$. \\
    $\mathcal{I}$ &$\models$ &$\neg P$ &$\text{ iff }$ &$\mathcal{I} \not\models P$. \\
    $\mathcal{I}$ &$\models$ &$P_{1} \vee P_{2}$ &$\text{ iff }$ &$\mathcal{I} \models P_{1} \text{ or } \mathcal{I} \models P_{2}$. \\
    $\mathcal{I}$ &$\models$ &${\bigcirc}P$ &$\text{ iff }$ &$k < j \text{ and } (\sigma,k+1,j) \models P$. \\
    $\mathcal{I}$ &$\models$ &$(P_{1},\ldots,P_{m}) \, \text{$prj$} \, P$ &$\text{ iff }$ &$\text{there exists integers $k=r_{0} \leq r_{1} \leq \cdots \leq r_{m} \preceq j$ such that}$ \\
    \multicolumn{5}{l}{$\text{$(\sigma,r_{0},r_{1}) \models P_{1}$, $(\sigma,r_{l-1},r_{l}) \models P_{l}$}(\text{for $1 < l \leq m)$,}$ and $(\sigma',0,|\sigma'|) \models P$ $\text{for one of the $\sigma'$:}$} \\
    %\multicolumn{5}{l}{and $(\sigma',0,|\sigma'|) \models P$ $\text{for one of the $\sigma'$:}$} \\
    \multicolumn{5}{l}{$\text{(a) $r_{m} < j$ and $\sigma'=\sigma \downarrow (r_{0},\ldots,r_{m}) \cdot \sigma_{(r_{m}+1 .. j)}$}$} \\
    \multicolumn{5}{l}{$\text{(b) $r_{m} = j$ and $\sigma'=\sigma \downarrow (r_{0},\ldots,r_{h})$}$ $\text{for some $0 \leq h \leq m$.}$} \\
    $\mathcal{I}$ &$\models$ &$P^{*}$ &$\text{ iff }$ &$\text{there are finitely many $r_{0},\ldots,r_{n} \in N_{\omega}$ such that}$ \\
    \multicolumn{5}{l}{$k=r_{0} \leq r_{1} \leq \cdots \leq r_{n-1} \preceq r_{n}=j (n \geq 0)$ and $\text{$(\sigma,r_{0},r_{1}) \models P$ and for all $1 < l \leq n$}$} \\
    \multicolumn{5}{l}{$(\sigma,r_{l-1},r_{l}) \models P$; or there are infinitely many integers $k=r_{0} \leq r_{1} \leq r_{2} \leq \cdots$} \\
    \multicolumn{5}{l}{such that $\lim\limits_{i\to\infty}r_{i}=\omega$ and $(\sigma,r_{0},r_{1}) \models P$ and for all $l > 1$ $(\sigma,r_{l-1},r_{l}) \models P.$}
\end{longtable}
}

A formula $P$ is satisfied over an interval $\sigma$, written $\sigma \models P$, if $(\sigma,0,|\sigma|) \models P$ holds. When $\sigma \models P$ holds for some interval $\sigma$, we say that formula $P$ is satisfiable. A formula $P$ is valid, denoted by $\models P$, if $\sigma \models P$ holds for all $\sigma$. Also we have the following derived formulas:
\begin{longtable}{rclrclrcl}
    \normalsize{$\varepsilon$} &\normalsize{$\overset{\text{def}}{=}$} &\normalsize{$\neg \bigcirc true$} &\quad \normalsize{$P_{1};P_{2}$} &\normalsize{$\overset{\text{def}}{=}$} &\normalsize{$(P_{1},P_{2}) \text{ $prj$ } \varepsilon$} &\normalsize{$P^{+}$} &\normalsize{$\overset{\text{def}}{=}$} &\normalsize{$P ; P^{*}$} \\
    \normalsize{$\lozenge P$} &\normalsize{$\overset{\text{def}}{=}$} &\normalsize{$true ; P$} & \normalsize{$\square P$} &\normalsize{$\overset{\text{def}}{=}$} &\normalsize{$\neg \lozenge \neg P$} & \normalsize{$\bigcirc^{n} P$} &\normalsize{$\overset{\text{def}}{=}$} &\normalsize{$\bigcirc (\bigcirc^{n-1} P), n \geq 1$}
\end{longtable}

Note that we use a finite set of natural numbers to denote locations. The main reason for this is that we want to preserve the decidability of PPTL$^{\tiny\mbox{SL}}$ while at the same time expressing more recursive heap properties. SL allows existential quantifiers, hence we can define properties about linked list using them. However, PPTL$^{\tiny\mbox{SL}}$ will be undecidable if the set of locations is infinite. Another way is to drop existential quantifiers in SL and keep the locations infinite. We will probably obtain a decidable logic but it is unable to describe complex heap properties.

\subsection{Specifying Heap Evolution Properties with PPTL$^{\tiny\mbox{SL}}$}

Consider the following C-like program, that first creates a linked list of some certain length (left part), then reverses its reference direction (right part). NULL is a macro for zero.
\begin{Verbatim}[commandchars=\\\{\},codes={\catcode`$=3\catcode`^=7},numbersep=-15pt,xleftmargin=0cm]
    \small{struct Node \{}
      \small{struct Node *next;}
    \small{\};}
    \small{function cre_rev() \{}
      \small{Node *x, *y, *t; int cnt := 0;}       \small{y := NULL;}
      \small{x := NULL;}                           \small{while(x != NULL) \textcolor{red}{\{\, \textcircled{\small{3}} \,\}} \{}
      \small{while(cnt < 100) \{}                       \small{t := x->next;}
        \small{t := new(Node);}                      \small{x->next := y;}
        \small{t->next := x;}                        \small{y := x;}
        \small{x := t;}                              \small{x := t;}
        \small{cnt := cnt+1;}                      \small{\}} \textcolor{red}{\{\, \textcircled{\small{2}} \,\}}
      \small{\}} \textcolor{red}{\{\, \textcircled{\small{1}} \,\}}                          \small{\}}
\end{Verbatim}
\noindent %
Some of the important state assertions specified by SL are labeled with \textcircled{\small{1}}, \textcircled{\small{2}} and \textcircled{\small{3}}, respectively:
\begin{eqnarray}
    &&\textcircled{\small{1}} (x=0 \wedge emp) \vee ls(x, 0) \;\;\; \textcircled{\small{2}} (y=0 \wedge emp) \vee ls(y, 0) \;\;\; \textcircled{\small{3}} (\textcircled{\small{1}} \text{\scriptsize{\#}} \textcircled{\small{2}}) \nonumber
\end{eqnarray}
Properties of interest for this program include the temporal relations among these state assertions, for instance:

(1) Two events happen sequentially: the first one is to create a list whose head pointer is $x$ resulting in \textcircled{\small{1}}, and the second is to reverse the list such that the head pointer of the resulting list will be $y$ leading to \textcircled{\small{2}}. More precisely, this property integrates heap shape property with interval property. Heap shape property can be expressed by \textcircled{\small{1}} and \textcircled{\small{2}}, and interval property by chop connective ``$;$''. PPTL$^{\tiny\mbox{SL}}$ formula $\lozenge \textcircled{\small{1}} ; \lozenge \textcircled{\small{2}}$ can expresses this property. It means sometimes in the heap, there only exists a complete linked list whose head pointer is $x$, and later the list becomes reversed with $y$ being the head.

(2) After the list is created, $x$ and $y$ will point to distinct lists (represented by $\textcircled{\small{3}}$) that repeatedly holds for several times. This property integrates heap non-interference property with loop property which can be described by the formula $\lozenge ((\bigcirc^{4} (\textcircled{\small{3}}))^{*})$. The formal property is treated by ``$\text{\scriptsize{\#}}$'', and the latter by star connective ``*''. Note that the formula in this example has an assumption that each statement executes in a unit interval. Eventually $\textcircled{\small{3}}$ holds for several times during the execution of the list reversal sub-program.

We can see that (1) and (2) are typical heap evolution properties that can be expressed neither by separation logic nor by temporal logics \cite{Verifying Temporal Heap Properties Specified via Evolution Logic,Safety and Liveness in Concurrent Pointer Programs,Verification of Pointer Programs,Model Checking Dynamic Memory Allocation in Operating Systems,Reasoning About Sequences of Memory States}. However, we can easily and clearly express them with PPTL$^{\tiny\mbox{SL}}$.

\section{Isomorphic Relationship Between PPTL and PPTL$^{\tiny\mbox{SL}}$}\label{sec: formula relationship}

PPTL and PPTL$^{\tiny\mbox{SL}}$ are closely related in their syntax structures since the only difference is the state assertions. In this section, an isomorphic relationship between PPTL and PPTL$^{\tiny\mbox{SL}}$ is presented. To do so, as depicted in Fig.\ref{relation}, first, we reduce a PPTL$^{\tiny\mbox{SL}}$ formula to an equisatisfiable PPTL$^{\tiny\mbox{SL}}$ formula in a restricted form which is a strict subset of PPTL$^{\tiny\mbox{SL}}$ (referred to as \emph{restricted PPTL$^{\tiny\mbox{SL}}$}). Second, an isomorphic relationship is built between PPTL formulas and the restricted PPTL$^{\tiny\mbox{SL}}$ formulas according to their syntax structures. To take an example of formula isomorphism, PPTL formula $Q \equiv p ; q$ is isomorphic to restricted PPTL$^{\tiny\mbox{SL}}$ formula $P \equiv x=0 ; y=0$ since their syntax structures are the same except the atomic formulas. $Q$ will be changed into $P$ if $p$ is replaced with $x=0$ and $q$ with $q=0$, and vice versa.

\begin{figure}[htb]
\centering
\includegraphics[scale=0.3]{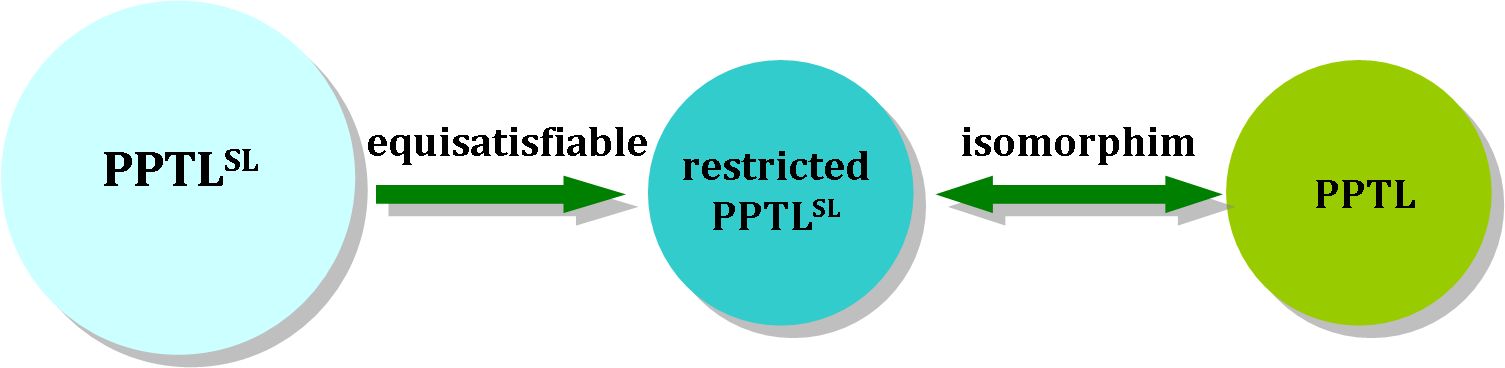}
\caption{\label{relation} The relationship between PPTL$^{\tiny\mbox{SL}}$ and PPTL}
\end{figure}

\subsection{Equisatisfiable Translation}

Two formulas are equisatisfiable if the first formula is satisfiable whenever the second is and vice versa. In other words, either both formulas are satisfiable or both are not. Two equisatisfiable formulas may have different models, provided they both have some or both have none. To start with, an equisatisfiable encoding for PPTL$^{\tiny\mbox{SL}}$ is proposed.

Calcagno et al. \cite{From Separation Logic to First-Order Logic} have already encoded the fragment of propositional separation logic into first-order logic. With the method in \cite{From Separation Logic to First-Order Logic}, we first encode the fragment of first-order separation logic (SL) into quantifier-free first-order logic so as to make the encoding closer to PPTL state formulas. The bounding property for satisfiability of state formulas is given which will be used to preserve the correctness of our translation. Bounded stacks and bounded heaps are defined as follows.
\begin{mydefinition}\label{def: Bounded states}
    {\rm A bounded stack written as $I_{s}[X]$ denotes the set of stacks such that $I_{s} \in I_{s}[X]$ iff $dom(I_{s}) = X$, where $X \subseteq Var$. A bounded heap written as $I_{h}[n]$ denotes the set of heaps such that $I_{h} \in I_{h}[n]$ iff $|dom(I_{h})| \leq n$, where $n \in \mathbb{N}$.} \qed
\end{mydefinition}

%A state can be represented as a pair of vectors, the first vector and the second one represent a stack and a heap respectively. On one hand, given a set $X = \{\, x_{1}, \ldots, x_{n} \,\}$ and a stack $I_{s} \in I_{s}[{X}]$, we define the representation of $I_{s}$ simply as the vector $((x_{1}, I_{s}(x_{1})), \ldots, (x_{n}, I_{s}(x_{n})))$ by the fixed ordering $(x_{1} < \cdots < x_{n})$ on variables as mentioned before. On the other hand,

Given a heap $I_{h} \in I_{h}[n]$, we will use a vector $c$ of $n$ pairs of values, $((c_{1,1}, c_{1,2}), \ldots,$ $(c_{n,1}, c_{n,2}))$, to represent that heap. If $c_{i,1} = 0$, the $i$-th pair does not represent an active heap cell, otherwise the cell is allocated at location $c_{i,1}$ and contains the value $c_{i,2}$. For example, $I_{h}[2]$ is a set of heaps which contains the heap of size one $I_{h} = \{\, (1, 2) \,\}$. Additionally, a vector allows the same location occurring more than once that should be avoid, e.g., $((1,2), (1,2))$ or $((1,2), (1,3))$ does not represent a \emph{valid} (or well-formed) heap. In order to overcome this problem, the partial function $vh_{n}$ is employed, $vh_{n} : (Val \times Val)^{n} \rightharpoonup I_{h}[n]$. In particular,
\begin{eqnarray}
    %&& vh_{n} : (Val \times Val)^{n} \rightharpoonup I_{h}[n] \nonumber \\
    && vh_{n}(c) =
    \begin{cases}
        \begin{split}
            &\mbox{Undef} \quad \mbox{if $\exists i, j : 1 \leq i, j \leq n, i \neq j, c_{i,1}=c_{j,1}, c_{i,1} \neq 0$ and $c_{j,1} \neq 0$}, \\
        \end{split} \\
        \{\, (c_{i,1}, c_{i,2}) \mid c_{i,1} \neq 0 \text{ and } 1 \leq i \leq n \,\}, \quad \mbox{otherwise}.
    \end{cases} \nonumber
\end{eqnarray}

Let $C$ denote a vector of pairs of variables, and $|C|$ indicates the number of variable pairs in $C$. If a vector $c$ with the same size is assigned to $C$, $C$ will also potentially represent a heap. In the following, the assertions about heaps in SL can be encoded as state formulas in PPTL$^{\tiny\mbox{SL}}$ in the following grammar
\begin{longtable}{lrcl}
    &$\phi_{s}$ &$\; ::= \;$ &$e_{1}=e_{2} \mid \neg \phi_{s} \mid \phi_{s_{1}} \vee \phi_{s_{2}}$
\end{longtable}

\noindent %
Given a vector of values $c = ((c_{1,1}, c_{1,2}), \ldots,$ $(c_{n,1}, c_{n,2}))$ and a vector of variables $C = ((C_{1,1}, C_{1,2}), \ldots, (C_{n,1}, C_{n,2}))$, we write $[C \Leftarrow c]$ to denote the pointwise assignment of values to the variables which is also considered as a set of pairs $\{\, (C_{i,1}, c_{i,1}), (C_{i,2}, c_{i,2}) \mid 1 \leq i \leq n \,\}$. The binary operation $\textcircled{\text{\scriptsize{\#}}}$ on vectors is in fact a formula defined in Definition \ref{def: decomposition}. It is adopted for capturing the meaning of separation conjunction \text{\scriptsize{\#}}.

\begin{mydefinition}[Vector Decomposition]\label{def: decomposition}
    {\rm For vectors of variables $C, C'$ and $C''$ such that $|C| = |C'| = |C''|$, we say $C$ is decomposed as $C'$ and $C''$, defined as
        \begin{equation}
            C = C' \textcircled{\text{\scriptsize{\#}}} C'' \overset{\text{def}}{=} \bigwedge_{i \in \{\, 1, \cdots, |C| \,\}} \left(
            \begin{split}
                \left( C'_{i,1}=C_{i,1} \wedge C''_{i,1}=0 \wedge C'_{i,2}=C_{i,2} \right) & \\
                \vee \left( C'_{i,1}=0 \wedge C''_{i,1}=C_{i,1} \wedge C''_{i,2}=C_{i,2} \right) &
            \end{split} \right) \nonumber
        \end{equation} \qed
    }
\end{mydefinition}

In the sequel, a set of pairs $D = \{\, (x_{1}, y_{1}), (x_{2}, y_{2}), \ldots \,\}$ with $\nexists (x,y), (x,z) \in D$ and $y \neq z$ is sometimes implicitly interpreted as a function. Conversely, a function $f$ can be interpreted as a set of pairs $\{\, (x, f(x)) \mid x \in dom(f) \,\}$. The standard notation $\bigvee_{i \in \{\, 1, \ldots, n \,\}} \phi_{s}$ is used to represent $\phi_{s}[1/i] \vee \cdots \vee \phi_{s}[n/i]$, and similarly for $\bigwedge_{i \in \{\, 1, \ldots, n \,\}} \phi_{s}$. As usual, the notation $fv(\phi)$ denotes the set of free variables occurring in $\phi$, which may be used to vectors, such as $fv(C)$.

\begin{lemma}\label{lemma: state formulas equisatisfiable}
    {\rm For any state formula $\phi$, variable vector $C$ and value vector $c$ where $|C| = |c| = n$, $(I_{s}, I_{h}) \in (I_{s}[fv(\phi)], I_{h}[n])$, $vh_{n}(c)=I_{h}$ and $fv(\phi) \cap fv(C) = \emptyset$, there exists a $\phi_{s}$ such that
    \begin{eqnarray}
        (I_{s}, I_{h}) \models_{_{SL}} \phi \quad \text{ iff } \quad (I_{s} \cup [C \Leftarrow c], \emptyset) \models_{_{SL}} \phi_{s} \nonumber
    \end{eqnarray}
    }
\end{lemma}
\begin{proof}
    We use a function $f$ to map a PPTL$^{\tiny\mbox{SL}}$ state formula $\phi$ to a state formula $\phi_{s}$ which is heap-free and has been defined before. The recursive translation $f(\phi,C)$ takes $\phi$ and $C$ as two parameters and produces a state formula $\phi_{s}$. The variables in $\phi$ and $C$ are always disjoint taking the form of two different syntactic categories.
    \begin{eqnarray}
        f(e_{1}=e_{2}, C) &\, \overset{\text{def}}{=} \,& e_{1}=e_{2} \nonumber \\
        %f'(nil=x_{1}, B) &\, \overset{\text{def}}{=} \,& q_{0,1} \nonumber \\
        %f'(x_{1}=x_{2}, B) &\, \overset{\text{def}}{=} \,& q_{1,2} \nonumber \\
        %&\, \vdots \,& \nonumber \\
        f(e_{1} \mapsto e_{2}, C) &\, \overset{\text{def}}{=} \,& \bigvee_{i \in \{\, 1, \cdots, |C| \,\}} \Bigg(
        \begin{split}
            & C_{i, 1} \neq 0 \wedge \bigwedge {}_{\substack{j \in \{\, 1, \cdots, |C| \,\} \\ i \neq j}} C_{j, 1}=0 \\
            & \wedge C_{i, 1}=e_{1} \wedge C_{i, 2}=e_{2}
        \end{split} \Bigg) \nonumber \\
        f(\neg \phi, C) &\, \overset{\text{def}}{=} \,& \neg f(\phi, C) \nonumber \\
        f(\phi_{1} \vee \phi_{2}, C) &\, \overset{\text{def}}{=} \,& f(\phi_{1}, C) \vee f(\phi_{2}, C) \nonumber \\
        f(\phi_{1} \text{\scriptsize{\#}} \phi_{2}, C) &\, \overset{\text{def}}{=} \,& \bigvee_{c_{2} \in Val^{2|C|}} \Big( \bigvee_{c_{1} \in Val^{2|C|}} \Big(
        \begin{split}
            & C=C' \textcircled{\text{\scriptsize{\#}}} C'' \wedge \\
            & f (\phi_{1}, C') \wedge f(\phi_{2}, C'')
        \end{split}
        \Big)[c_{1}/C'] \Big)[c_{2}/C''] \nonumber \\
        f(\exists x : \phi, C) &\, \overset{\text{def}}{=} \,& \bigvee_{v \in Val} f \left( \phi, C \right)[v/x] \nonumber
    \end{eqnarray}
    where both $C'$ and $C''$ are vectors with fresh variables, $[v/x]$ denotes the substitution of each occurrence of $x$ by $v$, and similarly for $[c/C]$ on vectors with pointwise substitution. One can draw the conclusion that $f(\phi, C)$ preserves the satisfaction of $\phi$ (similar to the proof of Theorem 1 given in \cite{From Separation Logic to First-Order Logic}). Therefore, the conclusion holds. \qed
\end{proof}

\begin{myexample}\label{example: singleton translation}
{\rm Consider the state formula $x \mapsto 0$, we can transform it into a state formula $\phi_{s}$ by $f$. Suppose $C = ((h_{1}, h_{1}'), (h_{2}, h_{2}'))$, the translation is
%\begin{eqnarray}
%    (B &\, = \,& \{\, x_{2}, x_{3} \,\}) \cap (fv(x_{1} \mapsto nil) = \{\, x_{1} \,\}) = \emptyset \nonumber \\
%    |B| &\, = \,& 1 < |x_{1} \mapsto nil| + |fv(x_{1} \mapsto nil)| = 1 + 1 = 2 \nonumber
%\end{eqnarray}
\begin{eqnarray}
     &&f(x \mapsto 0, C) \nonumber \\
     &=& f(x \mapsto 0, ((h_{1}, h_{1}'), (h_{2}, h_{2}'))) \nonumber \\
     &=& f((h_{1} \neq 0 \wedge h_{2} = 0 \wedge h_{1} = x \wedge h_{1}' = 0) \vee (h_{2} \neq 0 \wedge h_{1} = 0 \wedge h_{2} = x \wedge h_{2}' = 0)) \nonumber \\
     &=& (h_{1} \neq 0 \wedge h_{2} = 0 \wedge h_{1} = x \wedge h_{1}' = 0) \vee (h_{2} \neq 0 \wedge h_{1} = 0 \wedge h_{2} = x \wedge h_{2}' = 0) \nonumber
\end{eqnarray}
%\begin{eqnarray}
%     g \circ f(x \mapsto 0, B) &\, = \,& g (h_{1} \neq 0 \wedge h_{1} = x \wedge h_{1}' = 0) \nonumber \\
%     &\, = \,& \neg p_{1} \wedge p_{2} \wedge p_{3} \nonumber
%\end{eqnarray}
The rewritten result of the formula $(h_{1} \neq 0 \wedge h_{2} = 0 \wedge h_{1} = x \wedge h_{1}' = 0) \vee (h_{2} \neq 0 \wedge h_{1} = 0 \wedge h_{2} = x \wedge h_{2}' = 0)$ with any $C'$ is the same as that of $x \mapsto 0$ with $C$ since $f$ is a surjective. Observe that the result might be changed when the size of $C$ increases. However, this would not impact on the correctness of the translation, as only one heap cell is active.}
\end{myexample}

\begin{myexample}\label{example: sep translation}
    {\rm Consider the separation conjunction formula $x \mapsto 0 \text{\scriptsize{\#}} y \mapsto 0$, we can transform it into a formula $\phi_{s}$. Suppose $C = ((h_{1}, h_{1}'), (h_{2}, h_{2}'))$. According to the translation, we should select two fresh vectors $C'$ and $C''$ with same size of $C$, i.e., $C' = ((h_{3}, h_{3}'), (h_{4}, h_{4}'))$ and $C'' = ((h_{5}, h_{5}'), (h_{6}, h_{6}'))$.
    \begin{eqnarray}
        && \phi_{s_{1}} \nonumber \\
        &=& (C = C' \textcircled{\text{\scriptsize{\#}}} C'') \nonumber \\
        &=& ((h_{3}=h_{1} \wedge h_{5} = 0 \wedge h_{3}'=h_{1}') \vee (h_{3} = 0 \wedge h_{5}=h_{1} \wedge h_{5}'=h_{1}')) \wedge \nonumber \\
        && ((h_{4}=h_{2} \wedge h_{6} = 0 \wedge h_{4}'=h_{2}') \vee (h_{4} = 0 \wedge h_{6}=h_{2} \wedge h_{6}'=h_{2}')) \nonumber \\
        && \phi_{s_{2}} \nonumber \\
        &=&f(x \mapsto 0, C') \nonumber \\
        &=& f(x \mapsto 0, ((h_{3}, h_{3}'), (h_{4}, h_{4}'))) \nonumber \\
        &=& f((h_{3} \neq 0 \wedge h_{4} = 0 \wedge h_{3} = x \wedge h_{3}' = 0) \vee (h_{4} \neq 0 \wedge h_{3} = 0 \wedge h_{4} = x \wedge h_{4}' = 0)) \nonumber \\
        &=& (h_{3} \neq 0 \wedge h_{4} = 0 \wedge h_{3} = x \wedge h_{3}' = 0) \vee (h_{4} \neq 0 \wedge h_{3} = 0 \wedge h_{4} = x \wedge h_{4}' = 0) \nonumber \\
        && \phi_{s_{3}} \nonumber \\
        &=&f(y \mapsto 0, C'') \nonumber \\
        &=& f(y \mapsto 0, ((h_{5}, h_{5}'), (h_{6}, h_{6}'))) \nonumber \\
        &=& f((h_{5} \neq 0 \wedge h_{6} = 0 \wedge h_{5} = y \wedge h_{5}' = 0) \vee (h_{6} \neq 0 \wedge h_{5} = 0 \wedge h_{6} = y \wedge h_{6}' = 0)) \nonumber \\
        &=& (h_{5} \neq 0 \wedge h_{6} = 0 \wedge h_{5} = y \wedge h_{5}' = 0) \vee (h_{6} \neq 0 \wedge h_{5} = 0 \wedge h_{6} = y \wedge h_{6}' = 0) \nonumber \\
        && \phi_{s} \nonumber \\
        &=& f(x \mapsto 0 \text{\scriptsize{\#}} y \mapsto 0, C) \nonumber \\
        &\, = \,& f(x \mapsto 0 \text{\scriptsize{\#}} y \mapsto 0, ((h_{1}, h_{1}'), (h_{2}, h_{2}'))) \nonumber \\
        &\, = \,& \bigvee_{c_{2} \in Val^{2|C|}} \Big( \bigvee_{c_{1} \in Val^{2|C|}} (\phi_{s_{1}} \wedge \phi_{s_{2}} \wedge \phi_{s_{3}})[c_{1}/C'] \Big)[c_{2}/C''] \nonumber
    \end{eqnarray}
    }
\end{myexample}
%\begin{eqnarray}
%     g \circ f(x \mapsto 0, B) &\, = \,& g((h_{1} \neq 0 \wedge h_{1} = x \wedge h_{1}' = 0 \wedge h_{2}=0) \vee \nonumber \\
%     && (h_{2} \neq 0 \wedge h_{2} = x \wedge h_{2}' = 0 \wedge h_{1}=0)) \nonumber \\
%     &\, = \,& (\neg p_{1} \wedge p_{2} \wedge p_{3} \wedge p_{4}) \vee (\neg p_{4} \wedge p_{5} \wedge p_{6} \wedge p_{1}) \nonumber
%\end{eqnarray}
%where $p_{4}$ denotes $h_{2}=0$, $p_{5}$ denotes $h_{2}=x$ and $p_{6}$ denotes $h_{2}'=0$.

\subsubsection{How to choose $C$}

Let us restrict our attention to $C$ which should be chosen carefully. For example, if we assume $C$ in Example \ref{example: sep translation} with size one, i.e., $C = ((h_{1}, h_{1}'))$. It is impossible to find a suitable model because the heap is expected to have exactly two cells for the formula $x \mapsto 0 \text{\scriptsize{\#}} y \mapsto 0$. But the size of $|C|$ equals to one. Hence the rewritten formula is equivalent to $false$. In a word, $|C|$ is important in the translation which should not be too small or it may result in finding no satisfiable model for a given formula. In the following we provide a basic definition which is useful in choosing $C$. To take one example, the size of $e_{1} \mapsto e_{2}$ is one because, in order to decide whether it or its negation is satisfiable, it is enough to consider heaps with at most one allocated location. Consequently, when translating $e_{1} \mapsto e_{2}$, the size of parameter $C$ should be one or larger.

\begin{mydefinition}[Size of State Formula]\label{def: size of formulas}
    {\rm Given a state formula $\phi$, its size $|\phi|$ is defined by
        \begin{longtable}{rclrclrcl}
            $|e_{1}=e_{2}|$ &$=$ &$0$ &\qquad$|e_{1} \mapsto e_{2}|$ &$=$ &$1$ &\qquad$|\phi_{1} \vee \phi_{2}|$ &$=$ &$max(|\phi_{1}|, |\phi_{2}|)$ \\
            &&&$|\neg \phi|$ &$=$ &$|\phi|$ &\quad$|\phi_{1} \text{\scriptsize{\#}} \phi_{2}|$ &$=$ &$|\phi_{1}| + |\phi_{2}|$
        \end{longtable} \qed
    }
\end{mydefinition}

The quantifier does not appear in the above definition since it can be expanded into a disjunction formula. Roughly speaking, $|C| = |\phi| + |fv(\phi)|$ is just enough \cite{From Separation Logic to First-Order Logic} to bound the size of heaps that need to be considered. The previous lemma describes how the state formulas are encoded. We now in the position to translate full PPTL$^{\tiny\mbox{SL}}$ formulas. Before treating the translation, let us define \emph{restricted PPTL$^{\tiny\mbox{SL}}$} (RPPTL$^{\tiny\mbox{SL}}$ for short) formulas.
\begin{longtable}{lrcl}
    &$P_{s}$ &$\; ::= \;$ &$e_{1} = e_{2} \mid \neg P_{s} \mid P_{s_{1}} \vee P_{s_{2}} \mid {\bigcirc}P_{s} \mid (P_{s_{1}},\ldots,P_{s_{m}}) \, prj \, P_{s} \mid P_{s}^{*}$
\end{longtable}
\noindent %
It is easy to find that $\phi_{s}$ serves as state formulas in $P_{s}$. The translation $F$ defined below helps us to take charge of mapping a PPTL$^{\tiny\mbox{SL}}$ formula to a RPPTL$^{\tiny\mbox{SL}}$ formula. Also, this function preserves the satisfaction of $P$ where $C$ is a vector of variables, and $\varphi$ denotes $e_{1}=e_{2}, e_{1} \mapsto e_{2}, \phi_{1} \text{\scriptsize{\#}} \phi_{2}$ or $\exists x : \phi$.
\begin{longtable}{rclrcl}
    $F(\varphi,C)$ &\, $\overset{\text{def}}{=}$ \,& $f(\varphi,C)$ &\qquad$F(\neg P,C)$ &\, $\overset{\text{def}}{=}$ \,& $\neg F(P,C)$ \nonumber \\
    $F(P_{1} \vee P_{2},C)$ &\, $\overset{\text{def}}{=}$ \,& $F(P_{1},C) \vee F(P_{2},C)$ &$F(\bigcirc P,C)$ &\, $\overset{\text{def}}{=}$ \,& $\bigcirc F(P,C)$ \nonumber \\
    \multicolumn{6}{l}{$F((P_{1},\ldots,P_{m}) \; prj \; P_{0},C)$ \, $\overset{\text{def}}{=}$ \, $(F(P_{1},C),\ldots,F(P_{m},C)) \; prj \; F(P_{0},C))$} \nonumber \\
    $F(P^{*},C)$ &\, $\overset{\text{def}}{=}$ \,& $F(P,C)^{*}$ &&& \nonumber
\end{longtable}

Now, let us prove below a crucial result. Basically stating that translating $P$ to $P_{s}$ by the above encoding $F$ also produces an equisatisfiable result for $P$. It will turn out to be useful later on. Note that there is only a single vector $C$ when translating $P$, because different values can be assigned to $C$ for the sake of representing heap evolutions in an interval. Given an interval $\sigma$, $\sigma[(I_{s}, I_{h}) / (I_{s}^{i}, I_{h}^{i})]$ is an interval obtained by replacing the $i$-th state $(I_{s}^{i}, I_{h}^{i})$ with $(I_{s}, I_{h})$.

\begin{theorem}\label{theorem: PPTLSL formulas equisatisfiable}
    {\rm For any PPTL$^{\tiny\mbox{SL}}$ formula $P$, intervals $\sigma = \langle \ldots, (I_{s}^{i}, I_{h}^{i}),\ldots \rangle$ and $\sigma'$, set of variable vectors $C_{\sigma} = \{\, \ldots, C, \ldots \,\}$, set of value vectors $c_{\sigma} = \{\, \ldots, c_{i}, \ldots \,\}$ where $fv(P) \cap fv(C) = \emptyset$, $|C_{\sigma}| = |c_{\sigma}| = |\sigma|$, $\sigma' = \sigma[\cdots, (I_{s}^{i} \cup [C \Leftarrow c_{i}], \emptyset) / (I_{s}^{i}, I_{h}^{i}),$ $\cdots]$, $|C| = |c_{i}| = n$, $(I_{s}^{i}, I_{h}^{i}) \in (I_{s}[fv(P)], I_{h}[n])$, and $vh_{n}(c_{i}) = I_{h}^{i}$, for all $i$,
    \begin{eqnarray}
        (\sigma, 0, |\sigma|) \models P \quad \text{ iff } \quad (\sigma', 0, |\sigma'|) \models F(P,C) \nonumber
    \end{eqnarray}
    }
\end{theorem}
\begin{proof}
    The proof is based on a structural induction over $P$.

    \textbf{Case}: $P \equiv e_{1}=e_{2}$, $e_{1} \mapsto e_{2}$, $\phi_{1} \text{\scriptsize{\#}} \phi_{2}$ or $\exists x : \phi$

    $\Rightarrow:$ Suppose $(\sigma, 0, |\sigma|) \models P$. Since $P$ is a state formula, then $(I_{s}^{0}, I_{h}^{0}) \models_{_{SL}} P$. Since $(I_{s}^{0}, I_{h}^{0}) \in (I_{s}[fv(P)], I_{h}[n])$, $vh_{n}(c_{0})=I_{h}$ and $fv(P) \cap fv(C) = \emptyset$, by Lemma \ref{lemma: state formulas equisatisfiable}, $P$ is equisatisfiable to $f(P, C)$, i.e., $(I_{s}^{0} \cup [C \Leftarrow c_{0}], \emptyset) \models_{_{SL}} f(P,C)$. Hence, we have $(\sigma', 0, |\sigma'|) \models f(P,C)$. Furthermore, $F(P,C) = f(P, C)$ according to the definition of $F$. Thus, $(\sigma', 0, |\sigma'|) \models F(P,C)$.

    $\Leftarrow:$ Suppose $(\sigma', 0, |\sigma'|) \models F(P,C)$. By the definition of $F$, we have $F(P,C) = f(P,C)$. Moreover, since $f(P,C)$ is a state formula, $(I_{s}^{0} \cup [C \Leftarrow c_{0}], \emptyset) \models_{_{SL}} f(P,C)$. Since $(I_{s}^{0}, I_{h}^{0}) \in (I_{s}[fv(P)], I_{h}[n])$, $vh_{n}(c_{0})=I_{h}$ and $fv(P) \cap fv(C) = \emptyset$, by Lemma \ref{lemma: state formulas equisatisfiable}, $f(P, C)$ is equisatisfiable to $P$, i.e., $(I_{s}^{0}, I_{h}^{0}) \models_{_{SL}} P$. Therefore, $(\sigma, 0, |\sigma|) \models P$.

    Other cases are straightforward to be proved. \qed
\end{proof}

Recall that the size of $C$ corresponding to a heap size is required to be a bounded size when translating a state formula. But for a temporal formula $P$, there may be more than one state formula need to be considered at the same time. The max size should be selected. Concretely, the size of the vector $C$ for translating $P$ will be
\begin{eqnarray}
    |C|_{P} = max(\{\, |\phi| + |fv(\phi)| \;\;\big|\;\; \phi \text{ occurs in } P \,\}) \nonumber
\end{eqnarray}

For instance, for the formula $\bigcirc x = 0 \vee \square x \mapsto 0$, there exist two state formulas $x \mapsto 0$ and $x = 0$ in it. The size of $C$ for translating $P$ should be $max(\{\, 2, 1 \,\}) = 2$, which is the larger size for translating the two state formulas.

\begin{myexample}
{\rm Given a PPTL$^{\tiny\mbox{SL}}$ formula $P \equiv \bigcirc x = 0 \vee \square x \mapsto 0$, we can find a RPPTL$^{\tiny\mbox{SL}}$ formula $P_{s}$ which preserves the satisfaction of $P$ under the conditions presented in Theorem \ref{theorem: PPTLSL formulas equisatisfiable}. We choose the variable vector as $C=((h_{1}, h_{1}'), (h_{2}, h_{2}'))$.
\begin{eqnarray}
    &&F (\bigcirc x = 0 \vee \square x \mapsto 0, C) \nonumber \\
    &=&F(\bigcirc x=0, C) \vee F (\square x \mapsto 0, C) \nonumber \\
    &=&\bigcirc F(x=0, C) \vee \square F (x \mapsto 0, C) \nonumber \\
    &=&\bigcirc f(x=0, ((h_{1}, h_{1}'), (h_{2}, h_{2}'))) \vee \square f(x \mapsto 0, ((h_{1}, h_{1}'), (h_{2}, h_{2}'))) \nonumber \\
    &=&\bigcirc x=0 \vee \square \big((h_{1} \neq 0 \wedge h_{2} = 0 \wedge h_{1} = x \wedge h_{1}' = 0) \nonumber \\
    &&\vee (h_{2} \neq 0 \wedge h_{1} = 0 \wedge h_{2} = x \wedge h_{2}' = 0) \big) \nonumber
\end{eqnarray}}
\end{myexample}

In fact the above results enable us to only concentrate on RPPTL$^{\tiny\mbox{SL}}$ instead of PPTL$^{\tiny\mbox{SL}}$. RPPTL$^{\tiny\mbox{SL}}$ does not contain heap formulas, so it gives a more compact view of PPTL$^{\tiny\mbox{SL}}$. In the sequel, we will establish an isomorphism relationship between RPPTL$^{\tiny\mbox{SL}}$ and PPTL in a natural way so as to reuse the theory of PPTL.

\subsection{Isomorphism Relationship}

Let $L_{P_{s}}$ denote the set of all RPPTL$^{\tiny\mbox{SL}}$ formulas and $L_{Q}$ the set of all PPTL formulas. The second key step in our theory is to introduce a one-to-one relationship between $L_{P_{s}}$ and $L_{Q}$ with respect to their syntax structures.

\begin{lemma}\label{lemma: equation formulas and propositions}
    {\rm There exists a bijective relationship between atomic equation formulas of $P_{s}$ and atomic propositions of $Q$.}
\end{lemma}
\begin{proof}
    Let $Var = \{\, x_{0}, x_{1}, x_{2}, \ldots \,\}$ be the countable infinite set of variables. Assume the countable infinite set of propositions is $Prop = \{\, p_{0,1}, \ldots, p_{i, j}, \ldots,$ $q_{0,0}, \ldots, q_{i', j'}, \ldots \,\}$, where $0 \leq i' \leq n, 1 \leq j, 0 \leq j', 0 \leq i,$ and $i < j$. The function $g$ is defined as
    \begin{eqnarray}
        g(i=x_{j}) \overset{\text{def}}{=} q_{i,j} \qquad g(x_{i}=x_{j}) \overset{\text{def}}{=} p_{i,j}, i < j \nonumber
    \end{eqnarray}
    Obviously, $g$ is a bijective. Hence the conclusion holds. \qed
\end{proof}

It remains to establish structural isomorphism between RPPTL$^{\tiny\mbox{SL}}$ and PPTL. The next result gives another important step towards the development of our techniques. Before doing that, we define the formula isomorphism at the syntax structure level.

\begin{mydefinition}[Isomorphism]\label{def: isomorphism of formulas}
    {\rm Given a RPPTL$^{\tiny\mbox{SL}}$ formula $P_{s}$ and a PPTL formula $Q$, $P_{s}$ is said isomorphic to $Q$ (written as $P_{s} \cong Q$) if and only if \\
    (1) $P_{s} \equiv e_{1}=e_{2}$, $Q \equiv q$, $g(e_{1}=e_{2})=q$ ($g$ is defined in Lemma \ref{lemma: equation formulas and propositions}), or \\
    (2) $P_{s} \equiv \neg P_{s_{1}}$, $Q \equiv \neg Q_{1}$, $P_{s_{1}} \cong Q_{1}$, or \\
    (3) $P_{s} \equiv P_{s_{1}} \vee P_{s_{2}}$, $Q \equiv Q_{1} \vee Q_{2}$, $P_{s_{1}} \cong Q_{1}$, $P_{s_{2}} \cong Q_{2}$, or \\
    (4) $P_{s} \equiv \bigcirc P_{s_{1}}$, $Q \equiv \bigcirc Q_{1}$, $P_{s_{1}} \cong Q_{1}$, or \\
    (5) $P_{s} \equiv (P_{s_{1}},\ldots,P_{s_{m}}) \, prj \, P_{s_{0}}$, $Q \equiv (Q_{1},\ldots,Q_{m}) \, prj \, Q_{0}$, $P_{s_{i}} \cong Q_{i}$ for all $i$, or \\
    (6) $P_{s}\equiv P_{s_{1}}^{*}$, $Q \equiv Q_{1}^{*}$, $P_{s_{1}} \cong Q_{1}$.} \qed
\end{mydefinition}

Theorem \ref{theorem: isomorphism formulas} explains that there actually exists a bijective relationship between $L_{P_{s}}$ and $L_{Q}$ from the syntax equivalent point of view. It leads us to reuse the theory of PPTL for RPPTL$^{\tiny\mbox{SL}}$, especially the logic laws, decision procedure and the related definitions.

\begin{theorem}\label{theorem: isomorphism formulas}
    {\rm For any RPPTL$^{\tiny\mbox{SL}}$ formula $P_{s}$, there exists a PPTL formula $Q$ such that $P_{s} \cong Q$, and vice versa.}
\end{theorem}
\begin{proof}
    Given a formula $P_{s}$, a mapping $G : L_{P_{s}} \longrightarrow L_{Q}$ is constructed as
    \begin{longtable}{rclrcl}
        $G(e_{1}=e_{2})$ &$\overset{\text{def}}{=}$& $g(e_{1}=e_{2})$ &\quad $G(\neg P_{s})$ &$\overset{\text{def}}{=}$& $\neg G(P_{s})$ \\
        $G(\bigcirc P_{s})$ &$\overset{\text{def}}{=}$& $\bigcirc G(P_{s})$ &\quad $G(P_{s_{1}} \vee P_{s_{2}})$ &$\overset{\text{def}}{=}$& $G(P_{s_{1}}) \vee G(P_{s_{2}})$ \\
        \multicolumn{6}{l}{$G((P_{s_{1}},\ldots,P_{s_{m}}) \; prj \; P_{s_{0}}) \overset{\text{def}}{=} (G(P_{s_{1}}),\ldots,G(P_{s_{m}})) \; prj \; G(P_{s_{0}})$} \\
        $G(P_{s}^{*})$ &$\overset{\text{def}}{=}$& $G(P_{s})^{*}$ &&&
    \end{longtable}
    As expected, a formula $Q$ can be found such that $P_{s} \cong Q$ by Definition \ref{def: isomorphism of formulas}.

    Given a formula $Q$, a mapping $H : L_{Q} \longrightarrow L_{P_{s}}$ is constructed as
    \begin{longtable}{rclrcl}
        $H(q)$ &$\overset{\text{def}}{=}$& $g^{-1}(q)$ &\quad $H(\neg Q)$ &$\overset{\text{def}}{=}$& $\neg H(Q)$ \\
        $H(\bigcirc Q)$ &$\overset{\text{def}}{=}$& $\bigcirc H(Q)$ &\quad $H(Q_{1} \vee Q_{2})$ &$\overset{\text{def}}{=}$& $H(Q_{1}) \vee H(Q_{2})$ \\
        \multicolumn{6}{l}{$H((Q_{1},\ldots,Q_{m}) \; prj \; Q_{0}) \overset{\text{def}}{=} (H(Q_{1}),\ldots,H(Q_{m})) \; prj \; H(Q_{0})$} \\
        $H(Q^{*})$ &$\overset{\text{def}}{=}$& $H(Q)^{*}$ &&&
    \end{longtable}
    Hence a formula $P_{s}$ can be found such that $Q \cong P_{s}$ by Definition \ref{def: isomorphism of formulas}. \qed
\end{proof}

\section{Decision Procedure for PPTL$^{\tiny\mbox{SL}}$}\label{sec: Normal Form and NFG}

In the previous section, we prove an isomorphic relationship between PPTL$^{\tiny\mbox{SL}}$ and PPTL so as to reuse the theory of PPTL. We will sketch a decision procedure for the purpose of checking the satisfiability of PPTL$^{\tiny\mbox{SL}}$ formulas in this section. Due to space constraints, we do not present complete definitions and algorithms in the rest part of this section, they can be found in the papers \cite{A Decision Procedure for Propositional Projection Temporal Logic with Infinite Models,Complexity of Propositional Projection Temporal Logic with Star} with slight changes.

The decision procedure for checking the satisfiability of PPTL formulas relies heavily on a specific formula form called Normal Form. Informally, the normal form of a formula divides the formula into two rather intuitive parts: the present component and the future component, the former means the current interval ending point has been reached while the latter has the opposite meaning. Similarly, we can define normal form for RPPTL$^{\tiny\mbox{SL}}$ formulas since RPPTL$^{\tiny\mbox{SL}}$ is isomorphic to PPTL.
\begin{equation}\label{def: equation normal form}
    P_{s} \equiv \overset{n'}{\underset{j=1}{\bigvee}} (P_{e_{j}} \wedge \varepsilon) \vee \overset{n}{\underset{i=1}{\bigvee}} (P_{c_{i}} \wedge \bigcirc P_{i}') \nonumber
\end{equation}
where $P_{e_{j}}$ and $P_{c_{i}}$ are conjunctions composed of atomic equation formulas or their negations, and $P_{i}'$ is a general RPPTL$^{\tiny\mbox{SL}}$ formula.

%\begin{mydefinition}[semi-normal form]\label{def: semi-normal form}
%    {\rm Let $P_{s}$ be a restricted PPTL$^{\tiny\mbox{SL}}$ formula, $P_{s}$ is in normal form if $P_{s}$ has been rewritten as:
%    \begin{equation}\label{def: equation normal form}
%        P_{s} \overset{\text{def}}{=} \overset{n'}{\underset{j=1}{\bigvee}} (P_{e_{j}} \wedge \varepsilon) \vee \overset{n}{\underset{i=1}{\bigvee}} (P_{c_{i}} \wedge \bigcirc P_{i}') \nonumber
%    \end{equation}
%    where $P_{e_{j}}$ and $P_{c_{i}}$ are $true$ or state formulas of the form $\phi_{s}$, and each $P_{i}'$ is a general restricted PPTL$^{\tiny\mbox{SL}}$ formula. \qed}
%\end{mydefinition}

Using a very similar proof of Duan et al. \cite{A Decision Procedure for Propositional Projection Temporal Logic with Infinite Models,Complexity of Propositional Projection Temporal Logic with Star}, one can derive that any RPPTL$^{\tiny\mbox{SL}}$ formula is able to be written to its normal form since the logic laws can be inherited from PPTL.

We now give Algorithm \ref{NF algorithm} for transforming a PPTL$^{\tiny\mbox{SL}}$ formula to a normal form of its equisatisfiable RPPTL$^{\tiny\mbox{SL}}$ formula. The most important difference from the algorithm of Duan et al. lies in treating state formulas by using the translation formalized in the previous section. Other treatment on temporal connectives remain the same. In particular, the sub-algorithm CONF is used to transform a normal form into its complete normal form, while algorithm NEG is used to negate a complete normal form obtained from algorithm CONF. Algorithms PRJ and CHOP, respectively, are used to transform the formulas in projection and chop constructs to their normal forms. These algorithms are analogous to those given in \cite{A Decision Procedure for Propositional Projection Temporal Logic with Infinite Models} and \cite{Complexity of Propositional Projection Temporal Logic with Star}. Algorithm DNF equivalently rewrites a formula to its disjunction normal form.

\begin{algorithm} %算法开始
\small
\caption{\small{Algorithm for translating a PPTL$^{\tiny\mbox{SL}}$ formula to a normal form of its equisatisfiable RPPTL$^{\tiny\mbox{SL}}$ formula}} %算法的题目
\label{NF algorithm} %算法的标签
\textbf{Function} NF($F(P,C)$)
\begin{algorithmic}[1] %此处的[1]控制一下算法中的每句前面都有标号
%\REQUIRE $P$ is a PPTL$^{\tiny\mbox{SL}}$ formula %输入条件(此处的REQUIRE默认关键字为Require，在上面已自定义为Input)
%\ENSURE NF($\text{F}(P)$) computes an equivalent normal form for formula F$(P)$ %输出结果(此处的ENSURE默认关键字为Ensure在上面已自定义为Output)
\STATE \textbf{begin function}
\STATE \quad \textbf{case}
\STATE \quad \quad $P$ is $e_{1}=e_{2}$ or $e_{1} \mapsto e_{2}$ or $\phi_{1} \text{\scriptsize{\#}} \phi_{2}$ or $\exists x : \phi$: \textbf{return} $\text{DNF(}F(P,C)) \wedge \varepsilon \vee \text{DNF(}F(P,C)) \wedge \bigcirc true$;
%\STATE \quad \quad $P$ is $e_{1} \mapsto e_{2}$: \textbf{return} $\text{DNF}(\text{G}(P)) \wedge \varepsilon \vee \text{DNF}(\text{G}(P)) \wedge \bigcirc true$;
%\STATE \quad \quad $P$ is $\phi_{1} \text{\scriptsize{\#}} \phi_{2}$: \textbf{return} $\text{DNF}(\text{G}(P)) \wedge \varepsilon \vee \text{DNF}(\text{G}(P)) \wedge \bigcirc true$;
%\STATE \quad \quad $P$ is $\exists x : \phi$: \textbf{return} $\text{DNF}(\text{G}(P)) \wedge \varepsilon \vee \text{DNF}(\text{G}(P)) \wedge \bigcirc true$;
%\STATE \quad \quad $p$ is $\phi \wedge \varepsilon$: \textbf{return} $p$;
%\STATE \quad \quad $p$ is $\phi \wedge \bigcirc q$: \textbf{return} $p$;
\STATE \quad \quad $P$ is $P_{1} \vee P_{2}$: \textbf{return} $\text{NF(}F(P_{1},C)) \vee \text{NF(}F(P_{2},C))$;
%\STATE \quad \quad \textbf{case}
%%\STATE \quad \quad \quad $p$ is $\phi_{1} \vee \phi_{2}$: \textbf{return} $p$;
%\STATE \quad \quad \quad $q$ is a state formula and $r$ is a temporal formula: \textbf{return} $\text{SNF}(q) \vee \text{TNF}(r)$;
%\STATE \quad \quad \quad $r$ is is a state formula and $q$ is a temporal formula: \textbf{return} $\text{SNF}(r) \vee \text{TNF}(q)$;
%\STATE \quad \quad \quad $q$ and $r$ are temporal formulas: \textbf{return} $\text{TNF}(q) \vee \text{TNF}(r)$;
%\STATE \quad \quad \textbf{end case}
%\STATE \quad \quad $p$ is $q \wedge r$: \textbf{return} $\text{AND}(\text{NF}(q), \text{NF}(r))$;
%\STATE \quad \quad \textbf{case}
%%\STATE \quad \quad \quad $p$ is $\phi_{1} \wedge \phi_{2}$: \textbf{return} $p$;
%\STATE \quad \quad \quad $q$ is is a state formula and $r$ is a temporal formula: \textbf{return} $\text{AND}(\text{SNF}(q), \text{TNF}(r))$;
%\STATE \quad \quad \quad $r$ is is a state formula and $q$ is a temporal formula: \textbf{return} $\text{AND}(\text{SNF}(r), \text{TNF}(q))$;
%\STATE \quad \quad \quad $q$ and $r$ are temporal formulas: \textbf{return} $\text{AND}(\text{TNF}(q), \text{TNF}(r))$;
%\STATE \quad \quad \textbf{end case}
\STATE \quad \quad $P$ is $\neg P_{1}$: \textbf{return} $\text{NEG(CONF(NF($F(P_{1},C$))))}$;
\STATE \quad \quad $P$ is $\bigcirc P_{1}$: \textbf{return} $F(P,C)$;
%\STATE \quad \quad \textbf{case}
%%\STATE \quad \quad \quad $q$ is $\phi$: \textbf{return} $p$;
%\STATE \quad \quad \quad \textbf{return} $\text{NEG(CONF(TNF($q$)))}$;
%\STATE \quad \quad \textbf{end case}
%\STATE \quad \quad $p$ is $\square q$: \textbf{return} $\text{TNF}(q \wedge \varepsilon) \vee \text{TNF}(q \wedge \bigcirc \square q)$;
%\STATE \quad \quad $p$ is $\lozenge q$: \textbf{return} $\text{TNF}(q \vee \bigcirc \lozenge q)$;
\STATE \quad \quad $P$ is $P_{1} ; P_{2}$: \textbf{return} $\text{CHOP}(F(P,C))$;
\STATE \quad \quad $P$ is $(P_{1},\ldots,P_{m}) \text{ $prj$ } P_{0}$: \textbf{return} $\text{PRJ}(F(P,C))$;
\STATE \quad \quad $P$ is $P_{1}^{*}$: \textbf{return} $\varepsilon \vee \text{CHOP}(F(P_{1},C) ; F(P_{1},C)^{*})$;
\STATE \quad \textbf{end case}
\STATE \textbf{end function}
\end{algorithmic}
\end{algorithm}

Analogous to PPTL, RPPTL$^{\tiny\mbox{SL}}$ has its normal form which is useful for constructing a graph structure that explicitly characterizes the models of the corresponding formula. The graph structure, called Normal Form Graph (NFG), is constructed according to the normal form. For a RPPTL$^{\tiny\mbox{SL}}$ formula $P_{s}$, the NFG of $P_{s}$ is a directed graph, $G=(CL(P_{s}),EL(P_{s}))$, where $CL(P_{s})$ denotes the set of nodes and $EL(P_{s})$ denotes the set of edges in the graph. In $CL(P_{s})$, each node is specified by a formula in PPTL$^{\tiny\mbox{SL}}$, while in $EL(P_{s})$, each edge is a directed arc labeled with a state formula $P_{e}$ from node $P_{s}$ to node $P_{s}'$ and identified by a triple, $(P_{s}, P_{e}, P_{s}')$. In short, the NFG of $P_{s}$ can be built by a recursive approach.

As an example, consider the PPTL$^{\tiny\mbox{SL}}$ formula $P \equiv \bigcirc x=0 \vee \square x \mapsto 0$. We first translate $P$ to its equisatisfable formula $F(P, C)$ with $F$ and $C$, then the NFG of $F(P, C)$ can be constructed as shown in Fig.\ref{NFG}. The edges are labeled in red and the nodes in black.

\begin{figure}[htb]
\centering
\includegraphics[scale=0.73]{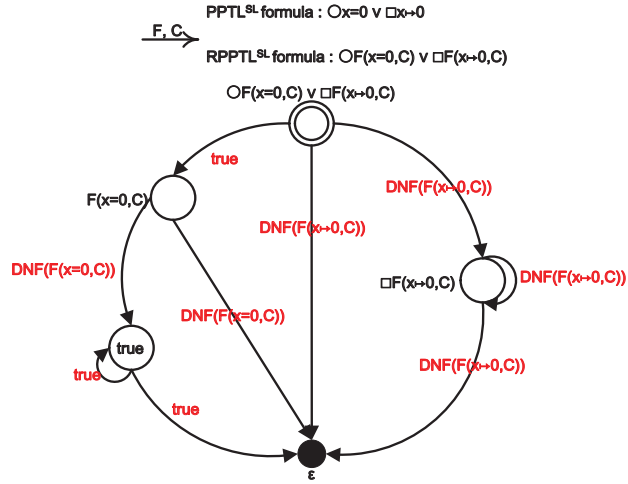}
\caption{\label{NFG} Example NFG}
\end{figure}

The edges labeled by state formulas which is unsatisfiable should be removed from an NFG. A finite path from the root node to the $\varepsilon$ node in the NFG of the formula corresponds to a finite model of the formula while an infinite path emanating from the root corresponds to an infinite model of the formula. There exists several finite or infinite path in Fig.\ref{NFG}. For instance, $\bigcirc F(x=0, C) \vee \square F(x \mapsto 0, C), \text{DNF}(F(x \mapsto 0, C)), \varepsilon$ is a finite path, and $\bigcirc F(x=0, C) \vee \square F(x \mapsto 0, C), true, F(x=0, C), \text{DNF}(F(x = 0, C)), true, \ldots$ is an infinite path. Therefore, the formula is satisfiable. Based on NFG, a decision procedure for checking satisfiability of PPTL$^{\tiny\mbox{SL}}$ formulas can be obtained similar to the one presented in \cite{A Decision Procedure for Propositional Projection Temporal Logic with Infinite Models,Complexity of Propositional Projection Temporal Logic with Star} for PPTL formulas.

\section{Conclusion}\label{sec: conclusion}

This paper integrates a decidable fragment of Separation Logic (SL) with Propositional Projection Temporal Logic (PPTL) to obtain a two-dimensional (spatial and temporal) logic PPTL$^{\tiny\mbox{SL}}$. The state formulas of PPTL$^{\tiny\mbox{SL}}$ are SL assertions, on top of which are the outer temporal connectives taken from PPTL. It is obvious to see that the two-dimensional logic marries the advantages of both, and it has the ability to relate consecutive configurations of the heap. In a word, it enables us to verify temporal properties of heaps.

Furthermore, in a general sense, another important contribution is that we also prove an isomorphism relationship between PPTL and PPTL$^{\tiny\mbox{SL}}$ formulas. This leads us to reuse PPTL theory to solve the satisfiability problem of PPTL$^{\tiny\mbox{SL}}$. In the future, a model checking approach by using PPTL$^{\tiny\mbox{SL}}$ as the specification language will be studied. We will possibly explore the unified model checking approach \cite{A Unified Model Checking Approach with Projection Temporal Logic} with PPTL$^{\tiny\mbox{SL}}$ as the specification language soon. The program is modeled by MSVL (Modeling Simulation and Verification Language) \cite{Framed Temporal Logic Programming} which is an executable logic programming language. In addition, to examine the entire approach, several big case studies will also be carried out.

%Some examples are given to show its applications. It is useful to specify temporal properties of heaps. We also prove a series of logic laws and a useful conclusion that any PPTL$^{\tiny\mbox{SL}}$ formula can be transformed into its normal form and LNFG. A normal form can be divided into two parts, one is the present component and the other is the future component. A LNFG precisely characterizes the semantics of a formula. Furthermore, a decision procedure for checking the decidability of PPTL$^{\tiny\mbox{SL}}$ will be investigated. In the future, a model checking approach using  PPTL$^{\tiny\mbox{SL}}$ as the specification language will be studied. In addition, to examine the approach, several case studies with larger examples are also required.

\end{document}